\documentclass{article}

\usepackage[english]{babel}

\usepackage[letterpaper,top=2cm,bottom=2cm,left=3cm,right=3cm,marginparwidth=1.75cm]{geometry}

\usepackage{amsmath}
\usepackage{amsmath,amssymb,amsfonts}
\usepackage{graphicx}
\newtheorem{theorem}{Theorem}

\newtheorem{remark}{Remark}
\newenvironment{proof}{{\indent \indent \it Proof:\quad}}{\hfill $\blacksquare$\par}

\usepackage{subcaption} % 引入 subcaption 宏包

\usepackage[colorlinks=true, allcolors=blue]{hyperref}

\title{Modeling and Optimizing Latency for Delayed Hit Caching with Stochastic Miss Latency}
\author{Bowen Jiang, Chaofan Ma\\
Shanghai Jiao Tong University}

\begin{document}
\maketitle

\begin{abstract}
Caching is crucial for system performance, but the delayed hit phenomenon, where requests queue during lengthy fetches after a cache miss, significantly degrades user-perceived latency in modern high-throughput systems. While prior works address delayed hits by estimating aggregate delay, they universally assume deterministic fetch latencies. This paper tackles the more realistic, yet unexplored, scenario where fetch latencies are stochastic. We present, to our knowledge, the first theoretical analysis of delayed hits under this condition, deriving analytical expressions for both the mean and variance of the aggregate delay assuming exponentially distributed fetch latency. Leveraging these insights, we develop a novel variance-aware ranking function tailored for this stochastic setting to guide cache eviction decisions more effectively. The simulations on synthetic and real-world datasets demonstrate that our proposed algorithm significantly reduces overall latency compared to state-of-the-art delayed-hit strategies, achieving a $3\%-30\%$ reduction on synthetic datasets and approximately $1\%-7\%$ reduction on real-world traces.
\end{abstract}

\section{Introduction}
Caching is a pivotal technique pervasively employed across diverse computing and networking systems, such as Content Delivery Networks (CDNs) \cite{cdn} and Mobile Computing environments \cite{storage}, to enhance system performance and user experience. Serving requests directly from the cache—a \textbf{hit}—yields low latency, whereas a cache \textbf{miss} necessitates fetching data from slower tiers (e.g., remote servers), thereby incurring significant delays. Traditional caching strategies span a wide spectrum: from classic heuristics such as Least Recently Used (LRU) \cite{LRU}, First-In-First-Out (FIFO) \cite{FIFO}, and Least Frequently Used (LFU) \cite{LFU}; to sophisticated methods including Adaptive Replacement Cache (ARC) \cite{ARC}, ADAPTSIZE \cite{adapt}, and Least Hit Density (LHD) \cite{lhd}; and extend to contemporary learning-based approaches \cite{lrb,deep,deep2}. These diverse algorithms have predominantly focused on maximizing the cache hit ratio. Implicitly or explicitly, much of this foundational work has operated under the assumption that the fetch latency upon a miss is negligible, particularly when compared to request inter-arrival times.

However, in modern high-throughput systems, this negligible-latency assumption often breaks down. Fetch latencies from remote servers can exceed $100ms$, while average request inter-arrival times may plummet to $1 \mu s$ \cite{mad}. Consequently, when an object miss occurs, subsequent requests for that same object arriving during the fetch interval cannot be served immediately; they must wait until the fetch completes. These waiting requests are termed \textbf{delayed hits}. With an evolving ratio between
system throughputs and latencies, the impact of delayed hits on overall user-perceived latency becomes increasingly significant.
Prior works addressing delayed hits, such as MAD \cite{mad}, LAC \cite{atc}, VA-CDH \cite{VA}, and CALA \cite{cala}, primarily focus on estimating the aggregate delay defined as the sum of the initial miss latency and subsequent delayed hit latencies.
Specifically, MAD \cite{mad} calculates the historical average aggregate delay (AggDelay), assuming all prior requests are missed. LAC \cite{atc} derives an analytical expression for the mean aggregate delay under Poisson arrivals. Viewing average estimates as potentially imprecise and simple upper bounds as too conservative, CALA \cite{cala} balances these by modeling aggregate delay as a weighted sum of AggDelay and the square of miss latency. VA-CDH proposes a novel ranking function considering the variance of aggregate delay \cite{VA}.
We note that all the previous works are typically under the assumption of deterministic miss latency. 

However, the assumption of deterministic fetch latency often deviates from reality, where network conditions and server load introduce significant stochasticity into fetch times. This paper tackles the delayed hit caching problem under this more realistic, yet challenging, condition of stochastic fetch latencies. Such randomness further increases the complexity and variability of the aggregate delay, making its variance a critical factor to consider alongside its mean for effective latency optimization. Building upon the variance-aware principles introduced in \cite{VA}, we provide a rigorous theoretical analysis deriving the exact mean and variance of aggregate delay under stochastic (specifically, exponentially distributed) miss latency. Based on this analysis, we develop the ranking function tailored for this stochastic miss latency to guide cache eviction decisions effectively.
In summary, our main contributions are as follows:
\begin{itemize}
    \item  We are the first work to model and analyze the delayed hit caching problem under the more realistic condition of stochastic miss latency. All previous works related to delayed hits assume that the miss latency is deterministic.
    \item  We derive exact analytical expressions for the mean and variance of the aggregate delay that occurs following a cache miss with stochastic fetch times.
     Leveraging this theoretical analysis, we obtain the variance-aware ranking function to guide cache evictions.
    \item We demonstrate through simulations on synthetic and real-world datasets that our proposed algorithm significantly reduces overall latency compared to state-of-the-art delayed-hit caching strategies, achieving a $3\%-30\%$ reduction average on synthetic datasets and $1\%-7\%$ reduction approximately on real-world traces.
\end{itemize}

\section{Problem Definition and Motivation}
\subsection{Problem Formulation}
We consider a cache of fixed capacity $C$ operating over a discrete time horizon $T$, serving requests $R_t$ for objects $i$ drawn from a universe of $N$ types at time $t$.
Every object $i$ is associated with a size $s_i$, constrained such that $\max_i s_i < C$.
Assuming that the arrival process for each object $i$ follows a Poisson process with a rate parameter $\lambda_i$.
Cache hits yield zero latency.
Cache misses initiate a fetch process from a remote server, taking a random time $Z_{i}$, and we model $Z_i$ with a probability density function(PDF) $g(z)$. 
For analytical clarity, we specifically assume $Z_i$ follows an exponential distribution characterized by a rate parameter $\mu_i$ (where the mean fetch time is $E[Z_i] = 1/\mu_i$, denoted as $z_i$), thus $g(z) = \mu_i e^{-\mu_i z}$ for $z > 0$.
Any requests for object $i$ arriving at time $t^{\prime}$ during its fetch interval $t^{\prime} \in (t,t+Z_i]$ become delayed hits, incurring latency equal to the remaining fetch time $Z_i-(t^{\prime}-$ $t).$ Upon fetch completion, an eviction policy may evict cached objects to accommodate the newly fetched object, aiming to minimize the sum of latencies experienced by all requests throughout the entire time horizon  $T$.

\begin{remark}
    In the analysis of stochastic miss latency, the exponential distribution, with parameter $\mu_i$, is adopted as the model. Its primary advantage lies in its mathematical tractability, greatly facilitating theoretical analysis and derivations within complex systems. The concise mathematical form of the exponential distribution, coupled with its unique memoryless property, enables the derivation of exact analytical expressions for key statistics of aggregate delay, such as its mean and variance. This provides an analytical foundation for understanding the impact of randomness on aggregate delay. While acknowledging that actual miss latency  may exhibit more complex behaviors (e.g., potential heavy-tailed phenomena) than the exponential distribution can capture, the exponential distribution, as a widely applied and fundamental probabilistic model, offers the necessary simplification and analytical feasibility to investigate this challenging problem and elucidate the underlying mechanisms of key factors.
\end{remark}

\subsection{The Aggregate Delay and Motivation}\label{latency}
% The delayed hit phenomenon decouples hit ratio maximization from latency minimization in caching systems. Consequently, accurately assessing the cost of a cache miss requires considering the aggregate delay ($D_i$). 
The aggregate delay is the initial miss latency plus the sum of latencies incurred by subsequent requests for the same object arriving during its fetch process, which is defined as follows at time $t$,
% For a stochastic fetch time $Z_i$, the aggregate delay for object $i$ following a miss at time $t$ is:
\begin{equation}\label{Aggdelay}
    D_i = Z_i + \sum_{t < t' \leq t + Z_i} (t+Z_i-t')\mathbb{I}(R_{t'} = i),
\end{equation}
where $\mathbb{I}(\cdot)$ is the indicator function and $R_{t'}$ denotes the request at time $t'$. 
However, since $D_{i}$ depends on future arrivals within the fetch interval $(t, t+Z_{i}]$, its exact value is unknown when the miss occurs. Estimating this aggregate delay is therefore a central challenge addressed by prior works \cite{mad,atc,cala,tan2025asymptotically}. 
% For example, MAD \cite{mad} calculates the historical average aggregate delay (AggDelay), assuming all prior requests were misses. LAC \cite{atc} derives an analytical expression for the expected aggregate delay under Poisson arrivals (rate $\lambda_{i}$). Viewing average estimates as potentially imprecise and simple upper bounds as too conservative, CALA \cite{cala} balances these by modeling aggregate delay as a weighted sum $D_{i}=(1-\gamma)$AggDelay$+\gamma z_{i}^{2}$, combining a historical estimate with an upper bound $(z_{i}^{2})$ via parameter $\gamma$. We also note related work \cite{tan2025asymptotically} incorporates eviction costs during fetch and considers bypass scenarios, aspects outside the scope of this paper.
% calculation in CaLa+ quantifies the penalty incurred when an aborted cache fetch forces requests to be served via bypassing. Since your specific scenario does not consider a bypass mechanism, the premise and consequences of EDF as defined in that text do not apply. Consequently, you do not need to consider or compare against CaLa+'s specific EDF cost term or its estimation method.
A key limitation of prior methods addressing delayed hits is their focus on estimating only the average aggregate delay, often neglecting its inherent variability. While the importance of variance has been noted in \cite{VA}, a clear illustration of its impact is beneficial. 
To demonstrate the necessity of considering the variance of aggregate delay, we present a toy example comparing two caching policies in Fig.\ref{toy}. We consider a cache of size 1, holding either object A or B, both with a fetch latency of $z=4ms$, and the request sequences are AAABAAABBBBAABBB. The initial cache is empty.
Policy 1 (Mean-based) prioritizes caching the object with the higher observed mean aggregate delay, while Policy 2 (Mean-Variance-based) prioritizes storing objects with a large sum of the mean and standard deviation of the aggregate latency.

\begin{figure}[h]
\centering
\includegraphics[width=3.5in]{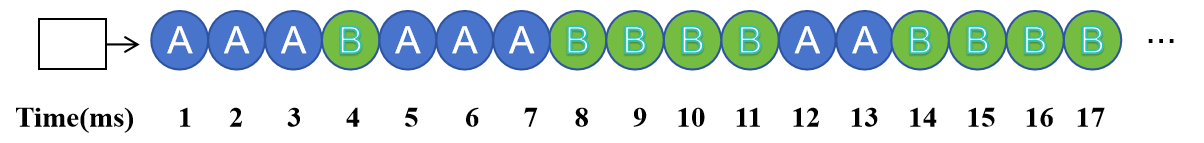}
% \captionsetup{justification=centering} % 设置图标题居中
\caption{A toy example to show the importance of variance in cache decision.
}
\label{toy}
\end{figure}
% Policy 1: The first arrival for A is a miss and yields $D_A=9$(t=1-3). When object B (from $t=4$ miss, resulting in $D_B=4$) arrives at $t=8$, the policy compares the observed mean aggregate delay $E_A=9$ to $E_B=4$ and retains A. 
% Consequently, the B request at $t=8$ misses, initiating a fetch and incurring a high aggregate delay $D_B=10$ (via delayed hits at $t=9-11$). When this second B arrives at $t=12$, the updated comparison ($E_A=9$ vs $E_B=7$) still favors keeping A cached. After subsequent A hits ($t=12-13$), the final four B requests ($t=14-17$) trigger another miss, accumulating a further $D_B=10$. This sequence of decisions leads to a final total latency of 33.
Policy 1: At time $t=1$, a request for A results in a cache miss with a latency of 4, causing subsequent requests at $t=2$ and $t=3$ to be delayed hits with latencies of 3 and 2, respectively. Following this, at $t=4$, a request for B incurs a miss with a latency of 4. The object A is retrieved and cached at $t=5$, leading to zero-latency hits for requests from $t=5$ to $t=7$. When B is retrieved at $t=8$, a comparison is made: A’s average aggregate delay (accumulated from $t=1-3$) is 9, while B’s (from the $t=4$ miss) is 4. Since A’s average latency is higher, A is retained in the cache. Consequently, requests for B from $t=8$ to $t=11$ result in a miss followed by delayed hits, incurring latencies of 4, 3, 2, and 1. At $t=12$, B is retrieved again; its new average latency becomes $(10 + 4) / 2 = 7$, while A’s remains 9. As A’s average is still higher, A is kept, resulting in hits with zero latency at $t=12$ and $t=13$. Subsequently, requests for B from $t=14$ to $t=17$ trigger another miss sequence with latencies of 4, 3, 2, and 1, respectively. Therefore, the total latency across this entire sequence under Policy 1 amounts to 33.

Policy 2: 
The cache behavior matches Policy 1 until $t=12$. At this point, a new rule based on the sum of average aggregate delay and standard deviation determines retention. A’s sum is 9 (avg 9, std 0), while B’s is 10 (avg 7, std 3). Since $10 > 9$, B is cached, causing A to miss at $t=12$ (latency 4) and have a delayed hit at $t=13$ (latency 3). B subsequently hits with zero latency at $t=14-15$. When A’s data arrives at $t=16$, the sums are re-evaluated: A is now 9 (avg 8, std 1), still less than B’s 10. B remains cached, resulting in zero-latency hits at $t=16-17$. This modified approach leads to a total latency of 30. This example demonstrates that considering variance leads to a better caching decision and reduces the overall latency.

\section{Characteristics of Aggregate Delay}\label{texingfenxi}
As previously demonstrated, incorporating both the mean and variance of aggregate delay into caching decisions enables more effective latency reduction. Motivated by this finding, we analyze the statistical properties including the distribution and variability of aggregate delay under the assumption of stochastic miss delays in this section.

\subsection{The distribution of aggregate delay}
% We apply the Law of Total Probability, conditioning on the random acquisition delay $Z_i$. 
% (\textbf{The PDF of $D_i$ under stochastic delay}).
Recall the PDF of the random miss latency is $g(z)$ and its mean value is $\mu_i=1/z_i$.
Because the object $i$ arrives according to a Poisson process with rate $\lambda_i$, thus the probability that the number of arrivals $N_i^z$ during the fetch period of length $z_i$ euqals $k$ is:
\begin{equation}
P(N_i^z = k) = \frac{e^{-\lambda_i z_i}(\lambda_i z_i)^k}{k!}, \quad k=0, 1, 2, \ldots
\label{eq:poisson_prob_detailed_en}
\end{equation}

Based on Theorem 1 in \cite{VA},
the PDF of $D_i$ can be obtained by integrating the conditional PDF $f(D_i | Z_i = z)$ over all possible values of the miss latency $z$, weighted by $g(z)$:
\begin{equation}\label{jifen}
\begin{split}
f(D_i)&= \int_0^\infty f(D_i | Z_i=z) g(z) dz\\
&= \int_0^\infty \left[ e^{-\lambda_i z} \delta(D_i - z) \right] \mu_i e^{-\mu_i z} dz+\\
& \quad \int_0^\infty \left[ \sum_{k=1}^{\infty} \frac{e^{-\lambda_i z}(\lambda_i z)^k}{k!} f_{D_i|k,z}(D_i) \right] \mu_i e^{-\mu_i z} dz,
\end{split}
\end{equation}
where $f_{D_i|k,z}(D_i) = \frac{1}{z^k (k-1)!} \sum_{j=0}^{\lfloor D_i/z \rfloor - 1} (-1)^j \binom{k}{j} (D_i - (j+1)z)^{k-1}$, $\lambda_i$ is the arrival rate for object $i$,  and $k$ is the total arrivals for object $i$ during the fetch process.
% The first part can be calculated as $\mu_i e^{-(\lambda_i + \mu_i) D_i}$ based on the property of Dirac $\delta$ function.
% It is extremely difficult to analytically solve the integration of the second part due to its variable of summation terms and complex integration domain. 
% The future work involves developing numerical techniques to calculate the integration for given $D_{i}$ and $\lambda_{i}$, and exploring effective approximation methods under specific conditions, such as large $k$ or small $\lambda_{i}$.
Using the property of the Dirac $\delta$ function $\left( \int \phi(z) \delta(z - a) dz = \phi(a) \right)$, the first part in \eqref{jifen} can be calculated as follows,
\begin{equation}
    \int_0^\infty \left[ e^{-\lambda_i z} \delta(D_i - z) \right] \mu_i e^{-\mu_i z} dz=\mu_i e^{-(\lambda_i + \mu_i) D_i}.
\end{equation}

% \begin{equation}
%    \begin{split}
%  \int_0^\infty \left[ e^{-\lambda_i z} \delta(D_i - z) \right] \mu_i e^{-\mu_i z} dz  & = e^{-\lambda_i D_i} \mu_i e^{-\mu_i D_i} \\
%  &  = \mu_i e^{-(\lambda_i + \mu_i) D_i}.
% \end{split}
% \end{equation}
Analytically solving the second integral in \eqref{jifen} presents major difficulties. Notably, the upper limit $\lfloor D_i/z \rfloor-1$ of the inner summation depends directly on the integration variable $z$, preventing straightforward integration as the number of terms changes. Furthermore, the effective integration range $z \in [D_i/(k + 1), D_i]$ varies with the summation index $k$, complicating the overall integration process.

\begin{remark}
    \textbf{(Hardness).}\emph{
Although we can formulate a formal integral expression for the PDF of the aggregate delay $D_i$ using the Law of Total Probability when the delay $Z_i$ follows an exponential distribution, obtaining its closed-form analytical solution for the general case $(k \geq 1)$ appears to be intractable due to the complexity of the resulting integral.
Potential future directions to address these challenges include employing numerical integration techniques to approximate the integral's value for specific parameters, likely after truncating the infinite series over $k$. Alternatively, analyzing the system using Laplace transforms ($L_{D_{i}}(s) = E[e^{-sD_i}]$) might offer a more structured approach \cite{Lap}, potentially simplifying the calculation of moments or enabling numerical inversion to approximate the probability density function.}
\end{remark}

\subsection{The mean and variance of aggregate delay}
The analysis of the mean and variance of aggregate delay under stochastic miss latency builds upon the corresponding results for deterministic miss latency. 
Therefore, we first present the mean and variance of aggregate delay for the deterministic scenario, as detailed in Theorem \ref{meanandvar}.

\begin{theorem}[in \cite{VA}]
\label{meanandvar} 
Assume object $i$ arrives according to a Poisson process with rate $\lambda_i$ and the miss latency is $z_i$. The mean and variance of $D_i$ are:
\begin{equation}
    E[D_i] = z_i \left(1 + \frac{\lambda_iz_i}{2}\right),\operatorname{Var}(D_i) = \frac{z_i^3 \lambda_i}{3}
\end{equation}

\end{theorem}

Under stochastic miss latency, the mean and variance of aggregate delays are shown in Theorem \ref{agg_delay_moments}.

\begin{theorem}\label{agg_delay_moments}
\emph{
Let the arrivals for object $i$ follow a Poisson process with rate $\lambda_i$. Let the miss latency $Z_i$, follow an exponential distribution with rate $\mu_i=1/z_i$, i.e., $Z_i \sim \text{Exp}(\mu_i)$.  Then, the mean and variance of the aggregated delay $D_i$ are given by:
\begin{align}
    E[D_i] &= z_i + \lambda_i{z_i^2} \label{eq:mean_agg_delay} \\
    Var(D_i) &= {z_i^2} +  6\lambda_iz_i^3+ 5 \lambda_i^2{z_i^4} \label{eq:var_agg_delay}
\end{align}}
\end{theorem}

\begin{proof}
The derivation proceeds by conditioning on the miss latency $Z_i \sim \text{Exp}(\mu_i), z_i = 1/\mu_i $ and applying the Law of Total Expectation for the mean and the Law of Total Variance for the variance.
We first derive the mean $E[D_i]$.
By the law of total expectation: $E[D_i] = E[ E[D_i | Z_i] ]$.

Based on Theorem \ref{meanandvar}, we have:
\begin{equation}\label{tiaojian}
     E[D_i | Z_i] = Z_i(1 + \lambda_i Z_i / 2)
\end{equation}

% First, we compute the conditional expectation $E[D_i | Z_i=z]$. This requires taking the expectation with respect to the conditional PDF $f(D_i | Z_i=z)$:
% \begin{equation}
%     E[D_i | Z_i=z] = \sum_{k=0}^\infty P(N_i^z=k | Z_i=z) E[D_i | N_i^z=k, Z_i=z]
% \end{equation}
% According to Lemma 1 in \cite{VA} , we know that given $k$ arrivals and an acquisition time $z$, the conditional expectation is $E[D_i | N_i^z=k, Z_i=z] = z(1 + k/2)$.
% Thus,
% \begin{equation}\label{tiaojian}
%     \begin{aligned}
%         E[D_i | Z_i=z] &= \sum_{k=0}^\infty \frac{e^{-\lambda_i z}(\lambda_i z)^k}{k!} \cdot z(1 + k/2) \\ &= z \sum_{k=0}^\infty \frac{e^{-\lambda_i z}(\lambda_i z)^k}{k!} + \frac{z}{2} \sum_{k=0}^\infty k \frac{e^{-\lambda_i z}(\lambda_i z)^k}{k!} \\ &= z \cdot 1 + \frac{z}{2} E[N_i^z | Z_i=z]  \\ &= z + \frac{z}{2} (\lambda_i z) = z(1 + \lambda_i z / 2)
%     \end{aligned}
% \end{equation}

Using Equation~\eqref{tiaojian} and $E[D_i] = E[ E[D_i | Z_i] ]$, we have:
\begin{equation}
    E[D_i] = E[ Z_i(1 + \lambda_i Z_i / 2) ] = E[Z_i + \lambda_i Z_i^2 / 2]
\end{equation}

Using the linearity of expectation:
$$E[D_i] = E[Z_i] + \frac{\lambda_i}{2} E[Z_i^2]$$
For $Z_i \sim \text{Exp}(\mu_i)$, we know that $E[Z_i] = 1/\mu_i$ and $Var(Z_i) = 1/\mu_i^2$. Thus, $E[Z_i^2] = Var(Z_i) + (E[Z_i])^2 = 1/\mu_i^2 + (1/\mu_i)^2 = 2/\mu_i^2$.
Substituting these values yields:
\begin{equation}
    E[D_i] = \frac{1}{\mu_i} + \frac{\lambda_i}{2} \left( \frac{2}{\mu_i^2} \right) = \frac{1}{\mu_i} + \frac{\lambda_i}{\mu_i^2} = z_i + \lambda_i{z_i^2}.
\end{equation}

Next, we derive the variance $Var(D_i)$.
By the law of total variance: $$Var(D_i) = E[Var(D_i | Z_i)] + Var(E[D_i | Z_i])$$
We need to compute two terms: $E[Var(D_i | Z_i)]$ and $Var(E[D_i | Z_i])$.

\textit{Computing $Var(D_i | Z_i=z)$:}
Based on Theorem \ref{meanandvar}, we have:
\begin{equation}
    Var(D_i | Z_i=z)=\frac{1}{3} \lambda_i z^3.
\end{equation}

Therefore,
$E[Var(D_i | Z_i)] = E\left[\frac{1}{3} \lambda_i Z_i^3\right] = \frac{\lambda_i}{3} E[Z_i^3]$.
For $Z_i \sim \text{Exp}(\mu_i)$, $E[Z_i^n] = n! / \mu_i^n$. Thus, $E[Z_i^3] = 3! / \mu_i^3 = 6 / \mu_i^3$.
$$E[Var(D_i | Z_i)] = \frac{\lambda_i}{3} \frac{6}{\mu_i^3} = \frac{2 \lambda_i}{\mu_i^3}$$

\textit{Computing the second term of the law of total variance: $Var(E[D_i | Z_i])$}
From Equation~\eqref{tiaojian}, we know that:
$$Var(E[D_i | Z_i]) = Var(Z_i(1 + \lambda_i Z_i / 2)) = Var(Z_i + \frac{\lambda_i}{2} Z_i^2)$$
Let $Y=E[D_i \mid Z_i]$. Using $Var(Y) = E[Y^2] - (E[Y])^2$, we know from the mean calculation that $E[Y] = E[Z_i + \frac{\lambda_i}{2} Z_i^2] = E[D_i] = \frac{1}{\mu_i} + \frac{\lambda_i}{\mu_i^2}$.
Therefore, we can obtain
% $E[Y^2] = E[(Z_i + \frac{\lambda_i}{2} Z_i^2)^2]
% = E[Z_i^2 + \lambda_i Z_i^3 + \frac{\lambda_i^2}{4} Z_i^4]
% =E[Z_i^2] + \lambda_i E[Z_i^3] + \frac{\lambda_i^2}{4} E[Z_i^4] 
% = \frac{2}{\mu_i^2} + \frac{6\lambda_i}{\mu_i^3} + \frac{6\lambda_i^2}{\mu_i^4}$.

\begin{equation}
    \begin{aligned}
        E[Y^2] &= E[(Z_i + \frac{\lambda_i}{2} Z_i^2)^2]\\
&= E[Z_i^2 + \lambda_i Z_i^3 + \frac{\lambda_i^2}{4} Z_i^4]\\
&=E[Z_i^2] + \lambda_i E[Z_i^3] + \frac{\lambda_i^2}{4} E[Z_i^4] \\ 
&= \frac{2}{\mu_i^2} + \frac{6\lambda_i}{\mu_i^3} + \frac{6\lambda_i^2}{\mu_i^4}.
    \end{aligned}
\end{equation}

Thus, we have,
\begin{equation}
    \begin{aligned}
    Var(E[D_i | Z_i]) &= \left( \frac{2}{\mu_i^2} + \frac{6\lambda_i}{\mu_i^3} + \frac{6\lambda_i^2}{\mu_i^4} \right) - \left( \frac{1}{\mu_i} + \frac{\lambda_i}{\mu_i^2} \right)^2 \\ 
    &= \frac{1}{\mu_i^2} + \frac{4\lambda_i}{\mu_i^3} + \frac{5\lambda_i^2}{\mu_i^4} 
    \end{aligned}
\end{equation}

Now we can compute the total variance $Var(D_i)$,
\begin{equation}
    \begin{aligned}
Var(D_i) &= E[Var(D_i | Z_i)] + Var(E[D_i | Z_i]) \\ &= \frac{2 \lambda_i}{\mu_i^3} + \left( \frac{1}{\mu_i^2} + \frac{4\lambda_i}{\mu_i^3} + \frac{5\lambda_i^2}{\mu_i^4} \right) \\ &= \frac{1}{\mu_i^2} + \frac{6 \lambda_i}{\mu_i^3} + \frac{5 \lambda_i^2}{\mu_i^4} \\
&= {z_i^2} +  6\lambda_iz_i^3+ 5 \lambda_i^2{z_i^4} .
    \end{aligned}
\end{equation}

\end{proof}

\begin{remark}
\textbf{(Parameter dependency).}
\emph{
Theorem \ref{agg_delay_moments} considers the aggregate delay where miss latency, $Z_{i}$, follows an exponential distribution with mean $z_{i}$. The mean aggregate delay grows faster with the average latency $z_{i}$ (due to the $z_{i}^{2}$ term) and linearly with arrival rate $\lambda_{i}$. Variance increases dramatically, involving higher powers of both $z_{i}$ and $\lambda_{i}$ (up to $\lambda_{i}^{2}z_{i}^{4}$), indicating high sensitivity to these parameters. Intuitively, variability now arises from two sources: the randomness of the latency $Z_{i}$ itself and the random arrivals interacting with this random latency. This dual randomness causes both mean and especially variance to react much more strongly to changes in system load $\lambda_{i}$ and average latency $z_{i}$.}
\end{remark}

\section{Our algorithm}\label{four}
Our algorithm builds upon the VA-CDH framework introduced in \cite{VA}, which aims to minimize latency in caching with delayed hits under deterministic miss latency. A core principle of VA-CDH is incorporating the variability of aggregate delay into eviction decisions. 

To achieve this, it employs a novel variance-aware ranking function to evaluate the aggregate delay of  object $i$:
\begin{equation}\label{rank}
f_i = \frac{E[D_i] + \omega \sigma[D_i]}{R_i s_i},
\end{equation}
where $E[D_i]$ and $\sigma[D_i]$ are the estimated mean and standard deviation of the aggregate delay $D_i$, respectively. $R_i$ is the estimated residual time until the next request for $i$, $s_i$ is its size, and $\omega > 0$ is a hyperparameter controlling the sensitivity to aggregate delay variability. A higher score $f_i$ indicates a higher priority for keeping object $i$ in the cache.

In this paper, we specifically address the scenario where the fetch latency $Z_i$ upon a miss is stochastic, modeled as an exponential distribution with mean $z_i = 1/\mu_i$. 
Based on the derived mean and standard deviation of the aggregate delay as presented in Theorem \ref{agg_delay_moments}, 
% of the aggregate delay under this assumption (as presented in Theorem \ref{agg_delay_moments} - ensure this theorem provides $E[D_i] = z_i + \lambda_i z_i^2$ and $Var(D_i) = z_i^2 + 6\lambda_i z_i^3 + 5\lambda_i^2 z_i^4$), 
the ranking function \eqref{rank} takes the specific form:
\begin{equation}\label{rank2}
f_i = \frac{(z_i + \lambda_i z_i^2) + \omega \sqrt{z_i^2 + 6\lambda_i z_i^3 + 5\lambda_i^2 z_i^4}}{R_i s_i},
\end{equation}
where the terms in the numerator represent the analytical $E[D_i]$ and $\omega\sigma[D_i]$ respectively.
Our algorithm maintains online estimates for the residual time $R_i$ (using LRU) and the arrival rate $\lambda_i$ (using the inverse of the mean inter-arrival time) similar to \cite{VA} . These parameters are estimated from recent request history within a sliding window of size $S$ to handle non-stationarity efficiently, avoiding the overhead of tracking the entire history. When an eviction is necessary, our algorithm calculates the rank $f_i$ for each cached object using its estimated $R_i$, $\lambda_i$, and the specialized ranking function \eqref{rank2}. Objects yielding the lowest rank scores are then evicted to make space for the newly fetched object.

\section{Simulation Results}\label{five}

\subsection{Methodology}
We consider the baseline algorithms, including  LRU \cite{LRU}, LAC \cite{atc}, VA-CDH \cite{VA},
CALA \cite{cala}, LHD\cite{lhd}, LRB\cite{lrb}, ADAPTSIZE \cite{adapt},
LHDMAD \cite{mad} and LRUMAD\cite{mad}.
We note that all of these baselines assume that the delay after a miss is a deterministic constant.

The performance of algorithm A is measured by the latency improvement relative to LRU\cite{cala}, which is defined as follows,
\begin{equation}
    \text { \small{Latency Improvement of } }\mathrm{A}=\frac{\text { \small{Latency(LRU)} }- \text {\small{ Latency(A)} }}{\text { \small{Latency(LRU)} }},
\end{equation} where $\text{Latency(LRU)}$ is the total latency of LRU algorithm, $\text{Latency(A)}$ is the total latency of algorithm A. 
% In subsection \ref{shiji}, we set $S=10K$ and $\omega=1$. 
In subsection \ref{hechengshiyan} and \ref{shiji}, we set $S=10K$ and $\omega=1$. 

% The impact of these parameters on the performance of our algorithm is discussed in subsection \ref{mingan}.

% Notably, the algorithms LAC \cite{atc}, CALA \cite{cala}, LHDMAD \cite{mad}, and LRUMAD \cite{mad} address the caching problems with delay hits, while ADAPTSIZE \cite{adapt}, LHD \cite{lhd}, and LRB \cite{lrb} focus on maximizing cache hit ratio without delayed hits.

% The performance of algorithm A is evaluated by its relative latency reduction compared to LRU \cite{cala}, calculated as:
% $(\text{Latency(LRU)} - \text{Latency(A)}) / \text{Latency(LRU)}$. Better performance is achieved with a higher improvement.

\subsection{Synthetic dataset}\label{hechengshiyan}
% We first validate the performance of VA-CDH using a synthesized dataset. In this dataset, request popularity follows a zipf distribution with a parameter of 1.2 across 100 distinct objects. Request sizes are uniformly distributed integers ranging from 1 MB to 100 MB. We evaluate the system under two separate request arrival scenarios: one following a Poisson process and the other a Pareto distribution. The latency incurred upon a cache miss is modeled as a constant delay plus a term proportional to the request size. The cache capacity is set to 500 MB and the total number of requests is 100K. The simulation results are shown in Fig.\ref{hecheng}.

\begin{figure}[h]
\centering
\includegraphics[width=2.5in]{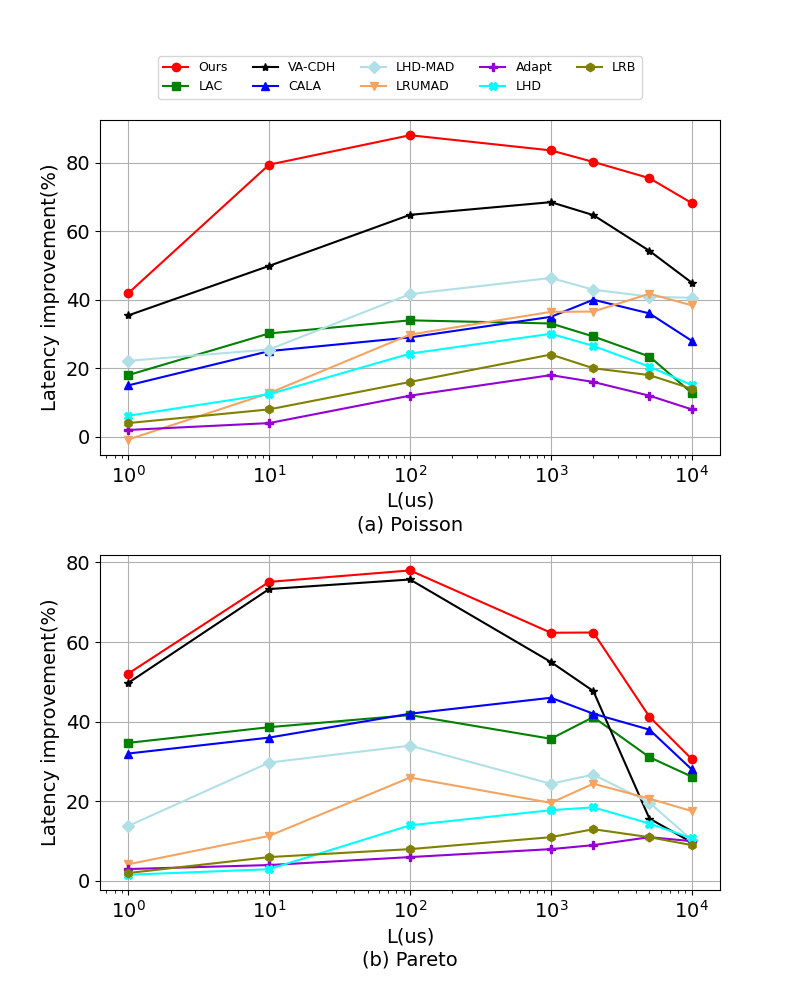}
% \captionsetup{justification=centering} % 设置图标题居中
\caption{Comparison of latency improvement between our algorithm and SOTAs under the
synthetic dataset($C=500\text{MB}$).
}
\label{hechengsuiji}
\end{figure}

The performance of the our algorithm is first evaluated using a synthetic dataset comprising 100,000 requests for 100 unique objects. Object popularity in this dataset follows a Zipf distribution, and object sizes are integers selected uniformly from the range $[1\text{MB}, 100\text{MB}]$. We simulate under two distinct arrival scenarios: one adhering to a Poisson process and another modeled by a Pareto distribution. 
The latency incurred upon a
cache miss is modeled as a constant delay $L$ plus
a component proportional to the object size.
 % The results, depicted in , demonstrate that VA-CDH consistently achieves superior performance compared to state-of-the-art (SOTA) algorithms across diverse miss latency configurations.
 Our algorithm can improve the performance by approximately $3\% -30\%$ on two different datasets as shown in Fig.\ref{hechengsuiji}.
 Importantly, this performance advantage persists even when arrivals follow the Pareto distribution, deviating from the Poisson model assumption used in our theoretical analysis.

% \begin{figure*}[t]
% \centering
%  % \captionsetup[subfigure]{font=small}
% {\includegraphics[width=7in]{figs/zitu2563.eps}}%
% \label{fig_first_case}
% \caption{Comparison of latency improvement between VA-CDH and SOTAs using a 256GB cache(deterministic delay).}
% \label{256queding}
% \end{figure*}

\subsection{Real world traces}\label{shiji}
%  We consider four real-world traces as follows:
% wiki18 \cite{lrb}, wiki19 \cite{lrb}, cloud \cite{cloud}, and YouTube \cite{you}. 
% We use a subset for simulation for the four datasets—Wiki2018, Wiki2019, Cloud, and YouTube. The total number of requests included in these datasets is 133,009, 155,015, 113,872, and 162,301, respectively. Figs. \ref{256} compare the latency improvement of our algorithm and state-of-the-art algorithms with different fetching latencies using a 256GB cache. 
% Our algorithm is generally consistently better than the state-of-the-art algorithms, as discussed in Section \ref{hechengshiyan} across various fetch latency settings. 
% Specially, our algorithm achieves a latency reduction improvement of approximately $1\%$ for the Wiki18 dataset compared to state-of-the-art algorithms. For the Wiki19 dataset, it provides an improvement of around $3\% $. For the Cloud dataset, the algorithm improves latency reduction by $2\%- 4\%$; for the YouTube dataset, the improvement ranges from $1\%- 6\% $.
% The key ingredient enabling further performance improvement is that our ranking function, by incorporating the variance of delayed hits, allows for a more effective prioritization of eviction candidates.

\begin{figure}[h]
\centering
\includegraphics[width=4in]{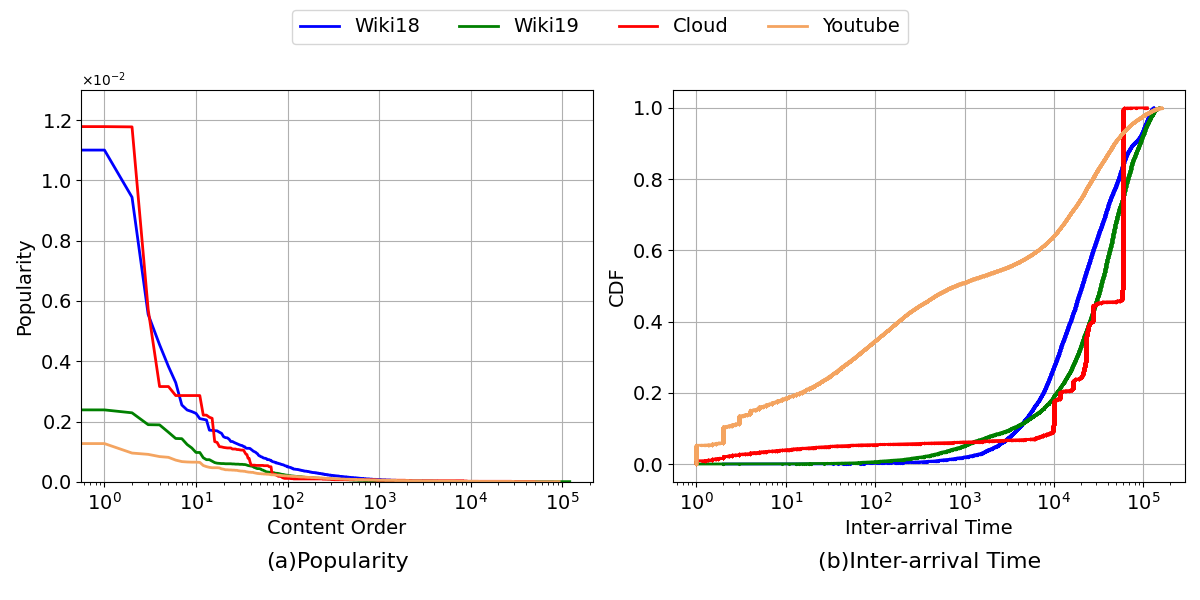}
% \captionsetup{justification=centering} % 设置图标题居中
\caption{The content popularity and average interval-time distributions of the
four real-world traces.}
\label{texing}
\end{figure}

\begin{figure}[h]
\centering
\includegraphics[width=2in]{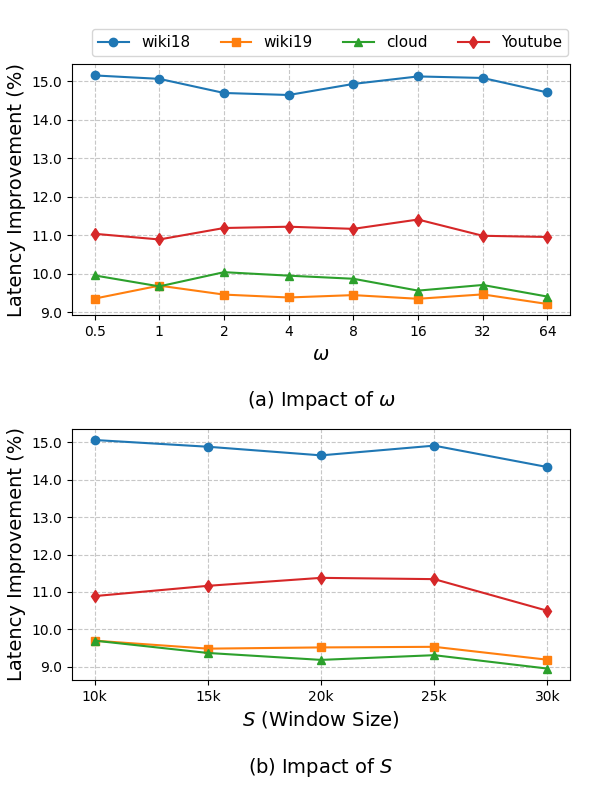}
% \captionsetup{justification=centering} % 设置图标题居中
\caption{The impact of hyperparameters on latency improvement.
}
\label{canshu}
\end{figure}

% We evaluate our algorithm using subsets of four real-world traces: Wiki2018 \cite{lrb}, Wiki2019 \cite{lrb}, Cloud \cite{cloud}, and YouTube \cite{you}. The simulated subsets contain approximately 133K, 155K, 114K, and 162K requests, respectively. 
% Fig. \ref {256} compares our algorithm's latency improvement (relative to LRU) against SOTA methods using a 256 GB cache across various fetch latency settings. As further detailed in section \ref{hechengshiyan}, our algorithm consistently outperforms the SOTA approaches under these conditions. Specifically, our algorithm achieves latency reductions of approximately $1\%-5\%$ on Wiki2018, around $1\%-4\%$ on Wiki2019, between $1\%-2\%$ on Cloud, and $3\%-7\%$ on YouTube. 
% This performance advantage is primarily attributed to our ranking function, which incorporates the randomness of miss latency to prioritize eviction candidates more effectively.

We further validate our algorithm on four real-world traces: Wiki2018 \cite{lrb}, Wiki2019 \cite{lrb}, Cloud \cite{cloud}, and YouTube \cite{you}. 
We present the object popularity and average inter-arrival time distributions of these four traces in Fig. \ref{texing}. 
Using a 256 GB cache and various fetch latency settings, Fig. \ref{256suiji} compares the latency improvement of our algorithm relative to LRU against state-of-the-art (SOTA) algorithms.

\begin{figure*}[h]
\centering
 % \captionsetup[subfigure]{font=small}
{\includegraphics[width=6in]{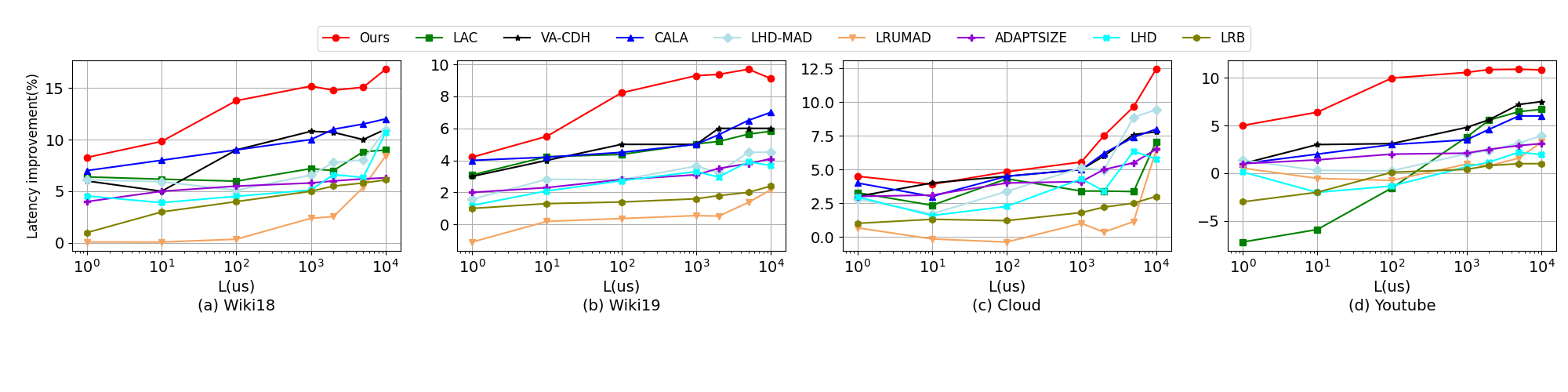}}%
\label{fig_first_case}
\caption{Comparison of latency improvement between our algorithm and SOTAs using a 256GB cache.}
\label{256suiji}
\end{figure*}

Our algorithm consistently outperforms these SOTA algorithms, achieving approximate latency reductions of $1\%-5\%$ on Wiki2018, $1\%-4\%$ on Wiki2019, $1\%-2\%$ on Cloud, and $3\%-7\%$ on YouTube. This performance advantage is primarily driven by our ranking function's effective incorporation of miss latency randomness to prioritize evictions.

% The distribution of average inter-arrival time illustrates that the Youtube and Cloud datasets exhibit more burstiness compared to the Wiki18 and Wiki19 datasets.

\subsection{Sensitivity Analysis}\label{mingan}

% \begin{figure}[h]
% \centering
% \includegraphics[width=1.9in]{figs/canshu2.eps}
% % \captionsetup{justification=centering} % 设置图标题居中
% \caption{The impact of hyperparameters on latency improvement(deterministic delay).
% }
% \label{canshu}
% \end{figure}

% The performance of our algorithm depends on two key hyperparameters: the sliding window size $S$ and the control parameter $\omega$. 
% The choice of the parameter $\omega$ involves a trade-off: if $\omega$ is too small, the algorithm may not sufficiently capture the impact of variance. Conversely, if $\omega$ is too large, the variance term can dominate the mean, potentially impairing the effectiveness of the priority determination. Similarly, the window size $S$ selection is critical: an overly small window prevents the algorithm from adequately learning request arrival characteristics, while an excessively large window hinders its ability to adapt effectively to changes in the request arrival process. Consequently, we perform a sensitivity analysis to evaluate their impact, and the simulation results are shown in Fig. \ref{canshu}.
% To investigate the impact of parameter $\omega$, we fix the sliding window size at $S=10K$. Conversely, to study the effect of the window size $S$, we fix $\omega=1$ and vary $S$. The constant part in miss latency is $L=5ms$.
% The experimental results indicate that our algorithm outperforms current SOTA methods across these different parameter settings. 
% This demonstrates the robustness of our algorithm.

We perform a sensitivity analysis, presented in Fig. \ref{canshu}, to assess our algorithm's performance dependency on its key hyperparameters: sliding window size $S$ and variance control parameter $\omega$. Both require careful tuning due to inherent trade-offs – $\omega$ balancing mean versus variance influence, and $S$ managing the compromise between historical learning effectiveness and responsiveness to changing request patterns. The analysis involves systematically varying $\omega$ (at $S=10K$) and $S$ (at $\omega=1$).
The delay $L$ in miss latency is $5ms$.
Across these varied settings, our algorithm consistently outperforms current STOA approaches, confirming its operational robustness.

\section{Conclusion}\label{six}
In this paper, we address the delayed hits in caching systems, specifically tackling the often-overlooked scenario of stochastic fetch latencies.
We provide the theoretical analysis under stochastic fetch latency, deriving exact expressions for the mean and variance of aggregate delay. Based on these findings, we have the variance-aware ranking function designed to guide eviction decisions by considering both the mean aggregate delay and its variability. The simulations validate our approach, demonstrating significant overall latency reductions compared to state-of-the-art delayed-hit algorithms.

\end{document}